\documentclass[11pt,letterpaper]{article}

\usepackage[margin=1in]{geometry}
\usepackage{stfloats}

\usepackage{hyperref}
\usepackage{colortbl}
\usepackage{paralist}
\usepackage{pdfpages}
\usepackage{enumitem}
\usepackage{float}
\usepackage{booktabs}
\usepackage[override]{cmtt}
\usepackage{color,xcolor}
\usepackage{graphicx,eso-pic}
\usepackage{boxedminipage}
\usepackage{url}
\usepackage{float}
\usepackage{caption}
\usepackage{subcaption}
\usepackage{xspace}
\usepackage{multirow}
\usepackage{amsmath,amsthm,amstext,amsfonts,latexsym}
\usepackage{amssymb}
\usepackage{wrapfig}
\usepackage[linesnumbered,ruled,vlined]{algorithm2e}
\usepackage[noend]{algpseudocode}
\usepackage{algorithmicx}
\usepackage{epstopdf}
\usepackage{mdframed}
\usepackage{cases}
\usepackage{mathrsfs}
\usepackage[capitalize]{cleveref}
\usepackage[compact]{titlesec}
\usepackage{tikz}
\usepackage{graphicx}
\usepackage{color}
\usetikzlibrary{positioning, shapes.geometric}
\usetikzlibrary{positioning,chains,fit,shapes,calc}

\usepackage{diagbox}

\AtBeginDocument{%
	}

\usepackage{paralist}
\usepackage{pdfpages}
\usepackage{enumitem}

\usepackage{scalerel}

\newcounter{note}[section]

\newcommand{\yuetodo}[1]{{\large\color{green}[Yue todo: #1]}}

\definecolor{yue}{rgb}{0.7, 0, 0}

\renewcommand{\yuetodo}[1]{}

\newcommand{\msf}[1]{\ensuremath{{\mathsf {#1}}}}

\newcommand{\mcal}[1]{\ensuremath{\mathcal {#1}}}

\newcommand{\ceil}[1]{\ensuremath{\left \lceil #1 \right \rceil}}

\definecolor{darkgreen}{rgb}{0,0.5,0}
\definecolor{lightblue}{RGB}{0,176,240}
\definecolor{darkblue}{RGB}{0,112,192}
\definecolor{lightpurple}{RGB}{124, 66, 168}
\definecolor{grey}{RGB}{139, 137, 137}
\definecolor{maroon}{RGB}{178, 34, 34}
\definecolor{green}{RGB}{34, 139, 34}
\definecolor{types}{RGB}{72, 61, 139}
\definecolor{gold}{rgb}{0.8, 0.33, 0.0}

\definecolor{darkgray}{gray}{0.3}

\definecolor{darkred}{rgb}{0.5, 0, 0}
\definecolor{darkgreen}{rgb}{0, 0.5, 0}
\definecolor{darkblue}{rgb}{0,0,0.5}

\newcommand\markx[2]{}

\newcommand{\Z}{\mathbb{Z}}

\newcommand{\ignore}[1]{}

\newcommand{\Geom}{\ensuremath{{\sf Geom}}}
\newcommand{\HS}{\ensuremath{{\sf HS}}}

\newcounter{task}

\newtheorem{theorem}{Theorem}[section]

\newtheorem{corollary}[theorem]{Corollary}
\newtheorem{fact}[theorem]{Fact}
\newtheorem{lemma}[theorem]{Lemma}

{
\newtheorem{definition}[theorem]{Definition}
\newtheorem{remark}[theorem]{Remark}

}

\newcounter{cnt:challenge}

\makeatletter
\def\moverlay{\mathpalette\mov@rlay}
\def\mov@rlay#1#2{\leavevmode\vtop{%
   \baselineskip\z@skip \lineskiplimit-\maxdimen
   \ialign{\hfil$\m@th#1##$\hfil\cr#2\crcr}}}
\newcommand{\charfusion}[3][\mathord]{
    #1{\ifx#1\mathop\vphantom{#2}\fi
        \mathpalette\mov@rlay{#2\cr#3}
      }
    \ifx#1\mathop\expandafter\displaylimits\fi}
\makeatother

\SetKwComment{Comment}{//}{}

\newcommand{\amount}{x} 
\newcommand{\price}{{p}}

\title{Indifferential Privacy: A New Paradigm and Its Applications to Optimal Matching in Dark Pool Auctions}

\author{Antigoni Polychroniadou\thanks{
J.P. Morgan AI Research, J.P. Morgan AlgoCRYPT CoE, USA.
\url{antigoni.polychroniadou@jpmorgan.com}}
\and
T.-H. Hubert Chan\thanks{University of Hong Kong. \url{hubert@cs.hku.hk}}
\and
Adya Agrawal\thanks{
J.P. Morgan Chase, India.
\url{adya.agrawal@jpmchase.com}}
}

\date{}

\begin{document}

\maketitle

\begin{abstract}

Public exchanges like the New York Stock Exchange and NASDAQ act as auctioneers in a public double auction system, where buyers submit their highest bids and sellers offer their lowest asking prices, along with the number of shares (volume) they wish to trade. The auctioneer matches compatible orders and executes the trades when a match is found. However, auctioneers involved in high-volume exchanges, such as dark pools, may not always be reliable. They could exploit their position by engaging in practices like front-running or face significant conflicts of interest—ethical breaches that have frequently resulted in hefty fines and regulatory scrutiny within the financial industry.

Previous solutions, based on the use of fully homomorphic encryption (Asharov et al., \mbox{AAMAS 2020}), encrypt orders ensuring that information is revealed only when a match occurs. However, this approach introduces significant computational overhead, making it impractical for high-frequency trading environments such as dark pools.

In this work, we propose a new system based on differential privacy combined with lightweight encryption, offering an efficient and practical solution that mitigates the risks of an untrustworthy auctioneer. Specifically, we introduce a new concept called Indifferential Privacy, which can be of independent interest, where a user is indifferent to whether certain information is revealed after some special event, unlike standard differential privacy. For example, in an auction, it's reasonable to disclose the true volume of a trade once all of it has been matched. Moreover, our new concept of Indifferential Privacy allows for maximum matching, which is impossible with conventional differential privacy.

\ignore{

\begin{itemize}
    \item Consider the shuffle model
    \item Multiple orders per user are submitted in units of 1 (always subtract from the real volume)
    \item Make the volume DP (obfuscate total volume from shuffler). Either add dummy or subtract from real volume.

    \item Use garbling to identify the users of the match at the end of the process (if both users agree to execute the match)
    \item Use order preserving encryption for the (ids +order id)
    \item If we add more volume is there an elegant solution? 

\item Formalize privacy leakage with respect to the matching algorithm; it should not know the exact volume.  But, what about other information such as prices?
    
\end{itemize}
}

\end{abstract}

\section{Introduction}
Dark pools are private trading venues designed for institutional investors to execute large trades anonymously, concealing details such as price and identities until after the transaction. Orders, which include the trade direction (buy or sell), volume, and price, are matched by the operator when they have opposite directions and compatible bid and ask prices. While they help prevent price swings from large orders, there are trust concerns. Operators may engage in front running, using insider knowledge of upcoming trades to execute their own trades first and profit from the price movement. Additionally, dark pool operators, often large financial institutions, may prioritize their own trades over clients', creating a conflict of interest. The lack of transparency also makes it difficult to detect manipulative practices, raising concerns about fairness. Several dark pool operators have been fined for misconduct, including misleading investors and failing to maintain proper trading practices. Notable examples include Barclays (\$70 million) and Credit Suisse (\$84.3 million) in 2016 for misrepresenting their dark pool operations, Deutsche Bank (\$3.7 million in 2017) for similar issues, ITG (\$20.3 million in 2015) for conflicts of interest, and Citigroup (\$12 million in 2018) for misleading clients about trade execution~\cite{fines}.

While dark pools aim to prevent the leakage of large orders, operators still gain privileged access to clients' hidden orders, creating potential for conflicts of interest or misuse of sensitive data. Recent research has focused on cryptographically protecting order information. These systems allow users to submit orders in encrypted form, enabling dark pool operators to compare orders without revealing their contents, only unveiling them when matches occur. This approach ensures greater security and mitigates risks associated with operator access to sensitive trade information. 

Asharov et al.~\cite{AsharovBPV20} introduced a secure dark pool model using Threshold Fully Homomorphic Encryption (FHE), which combines two cryptographic techniques: FHE, allowing computations on encrypted data without decryption, and Threshold Cryptography, where a secret (such as a decryption key) is split among multiple parties. In this approach, data is encrypted with a public key, and computations are performed on the encrypted data by an untrusted party. Decryption requires a threshold number of participants to combine their key shares. In Asharov et al.~\cite{AsharovBPV20} model, orders are encrypted under a public key, and the operator matches them directly on the encrypted data. Once a match is found, the orders are decrypted by the clients using their decryption shares, ensuring that no single entity, including the operator, can access the sensitive order details. Throughout the process, the orders remain encrypted, preventing any single party from accessing both the data and the decryption key. However, FHE is known for being computationally heavy, and adding a threshold mechanism increases the complexity. Optimizing this for real-time or large-scale applications is still an active research area. Moreover, the need for multiple parties to collaborate on decryption and sometimes computation can introduce significant communication overhead. As a result, the process takes nearly a full second to complete a single match, significantly impacting performance.

A recent development, Prime Match~\cite{polychroniadou2023prime}, introduces a solution to protect the confidentiality of periodic auctions run by a market operator. Prime Match allows users to submit orders in an encrypted form, with the operator comparing these orders through encryption and only revealing them if a match occurs. The auctions capture the trade direction (buy or sell) and the desired volume, but exclude price. Prime Match represents the first financial tool based on secure multiparty computation (MPC). In the high-stakes, highly competitive financial sector, MPC is gaining significant traction as a crucial enabler of privacy, with J.P. Morgan successfully deploying Prime Match in production.

However, in continuous double auctions without excluding the price, such as those in dark pools, where the computational complexity surpasses that of simpler periodic auctions like Prime Match, secure computation techniques fall short. They cannot support high-frequency trading within acceptable timeframes, limiting the feasibility of these methods for enhancing privacy in dark pools and the broader financial sector, where speed and efficiency are critical.

In this work, we pose the question: Can we achieve a privacy-preserving solution to dark pools with efficiency comparable to non-private dark pool protocols? We propose a system that combines differential privacy with encryption, providing a more efficient alternative to secure MPC and FHE. Differential Privacy achieves privacy by adding noise to the results of queries or computations on datasets. The level of noise is determined by a privacy parameter, which quantifies the trade-off between privacy and accuracy. By leveraging differential privacy, our dark pool approach ensures that individual orders are obfuscated while still allowing for effective matching. This method conceals the most critical aspect of dark pools—the volumes of orders—thus preserving the primary objective of dark pools, which is to hide large trades. 

In summary, integrating differential privacy with encryption offers a streamlined, efficient solution that balances privacy and computational feasibility, making it an attractive practical alternative to impractical methods like MPC and FHE. 

\subsection{Our Contributions:}

\paragraph{Problem Statement:}

The dark pool consists of $n$ agents (clients), and an operator who receives the orders from the clients. Each order takes one of two forms: (1) Buy Order: $({\sf buy}, \price, \amount)$, where $\amount$ is the quantity/volume, and $\price$ is the highest price the buyer is willing to pay per share.
(2) Sell Order: $({\sf sell}, \price, \amount)$, where $\price$ is the lowest price the seller is willing to accept per share. A buy order can be matched to a sell order if the buying
price is at least the selling price. Our objective is to design matching protocols that maximize the total number of matches while preserving user data privacy. Specifically, we aim to conceal the quantity $\amount$ of the orders during the process. Our key contributions are as follows:
\begin{enumerate}
    \item {\bf Practical Dark Pool Solution:} We propose an efficient solution for continuous double auctions (such as dark pools) that conceals both bid and ask quantities, effectively reducing reliance on trusted auctioneers while ensuring privacy guarantees.
    \item {\bf Novel Privacy Concept:} We introduce indifferential privacy, a new extension of differential privacy tailored to this context, which can be of independent interest with potential applicability beyond auctions and dark pools.
    \item {\bf Maximum Matching:} Our new notion of indifferential privacy allows us to achieve the optimal maximum matching which is impossible to achieve under conventional differential privcy~\cite{DBLP:journals/siamcomp/HsuHRRW16}. 
     \item {\bf Efficiency and Implementation: } Our system significantly outperforms previous privacy-preserving auction models, which often struggled with practicality and hindered their adoption in production. We show that our solution rivals non-private auction protocols in terms of performance, making it viable for real-world deployment. 
\end{enumerate}

\paragraph{High-Level Idea of our techniques:} To preserve privacy while matching buy and sell orders, each order is viewed as containing 
$\amount$ units, with users submitting 
$\amount+ noise$  orders to the server based on indifferential privacy. The orders are represented as nodes in a bipartite graph, with sell orders from sellers $S_i$ on the left and buy orders from buyers $B_i$ on the right. See Figure~\ref{fig:auction} for an example. The server ranks the nodes based on price $p$, where higher prices for buyers and lower prices for sellers are considered more favorable, referred to as "extreme" prices. The algorithm then constructs a bipartite graph, with edges connecting buy and sell nodes if the buying price meets or exceeds the selling price.

\begin{figure}[hbt!]
    \centering
    \includegraphics[width=0.6\columnwidth]{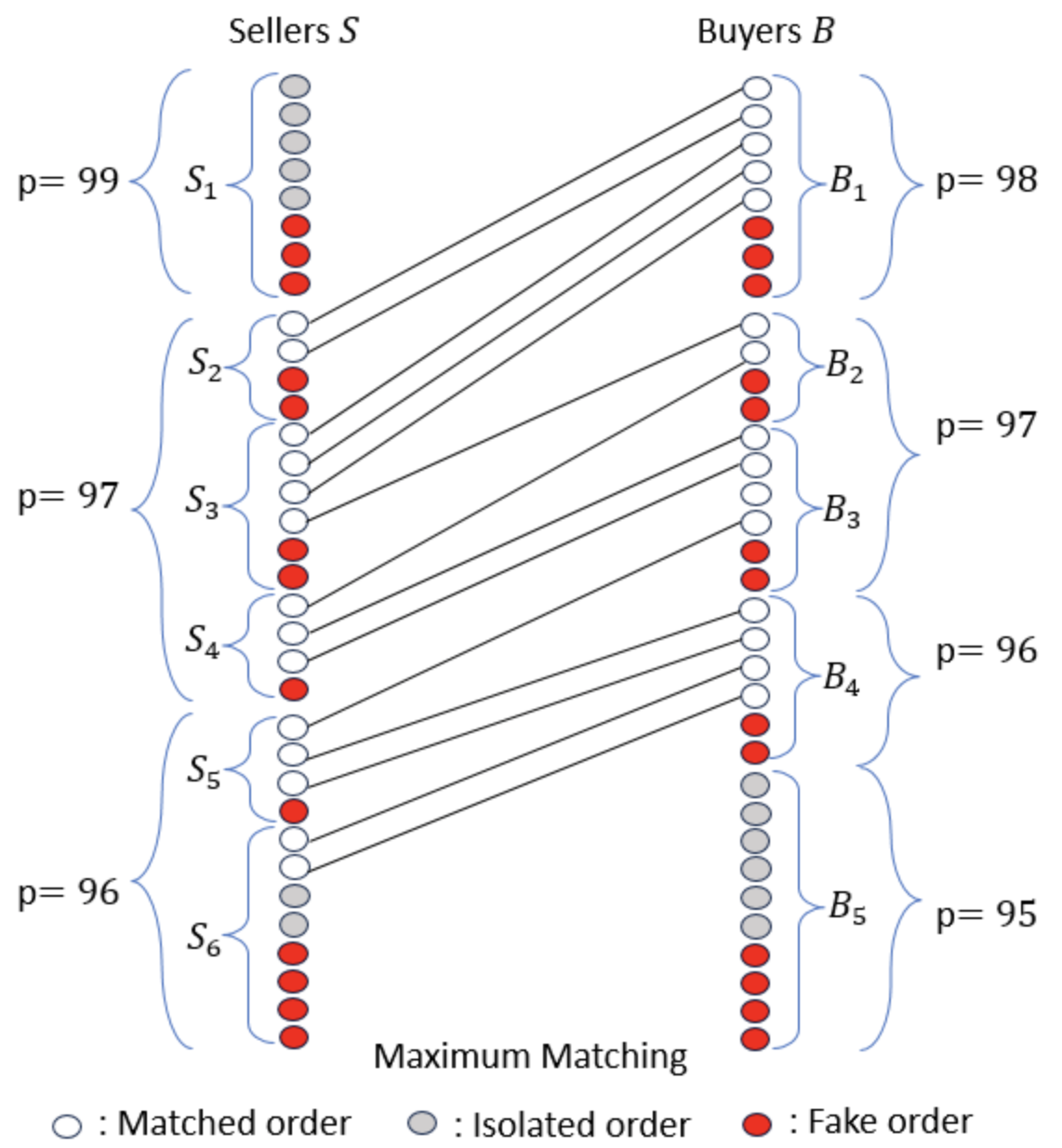}  
    \caption{\textbf{Maximum Matching Example.} \textnormal{The figure illustrates a bipartite graph and its maximum matching where nodes are first sorted by price, followed by the user's genuine orders and then their fake orders. This arrangement ensures optimal matching, maximizing the number of successful pairings according to algorithm~\ref{alg:matching}.}}
    \label{fig:auction}
\end{figure}

In particular, the algorithm constructs and maintains a bipartite graph, where edges exist between buy and sell nodes when their prices are compatible, i.e., the buying price is at least equal to the selling price. As matches are made, isolated nodes (the gray nodes in Figure~\ref{fig:auction})—those without neighbors, such as when their price cannot be met—are promptly removed from the graph. 
In the maximum matching problem, the goal is to identify a set of edges where no two edges share a node, thereby maximizing the total number of matches. As mentioned above, nodes with the most popular price on one side of the graph are naturally connected to nodes with the least popular price on the opposite side. These pairs, referred to as polar opposite nodes, form the basis of an iterative matching strategy that guarantees an optimal matching solution.

Our approach maintains this optimal matching even when users submit a number of noise-injected fake (the red nodes in Figure~\ref{fig:auction}) orders to obscure the true amounts of their order. To achieve this, we introduce  in-differential privacy (detailed in Section~\ref{sec:IDP}) and orders are not only sorted by price but also arranged such that each user’s true orders are followed by their corresponding fake orders. 

Unlike traditional differential privacy, our new in-differential privacy concept allows for selective disclosure of information post-match, aligning with realistic scenarios. For instance, in the context of a dark pool trading environment, it becomes acceptable to reveal the actual quantity of a trade once it has been fully matched, as it no longer poses a risk to privacy.
To establish this new privacy framework, we integrate graph refinement techniques, ensuring that the protocol not only protects user data but also facilitates an efficient and optimal matching process, even under the presence of noise.

\paragraph{Implementation:} We present an end-to-end implementation of our system. While previous FHE-based solutions processed fewer than one order per second, our system dramatically outperforms them, handling between 600 and 850 orders per second (according to Table~\ref{tab:comparison}), depending on the input volume. Furthermore, we provide an analysis of the overhead introduced by our privacy-preserving mechanism compared to the non-private version. While privacy inevitably incurs some cost and does not come for free, our system's overhead remains minimal and practical, making it highly suitable for high-frequency trading environments.

\subsection{Related Work}

The works of~\cite{CartlidgeSA19,CartlidgeSA21,MazloomDPB23} leverage secure multiparty computation (MPC) with multiple operators instead of a single one. While this method is computationally faster than fully homomorphic encryption (FHE), it comes with significant communication overhead due to the necessary interactions among multiple operators. Moreover, the practicality of this approach is limited, as most current dark pool systems operate with a single operator. Importantly, MPC can only guarantee the privacy of orders if the dark pool operators do not collude, raising concerns in scenarios where collusion is feasible. Given these challenges, it is crucial to focus on solutions that emphasize single-operator architectures, which can streamline communication and enhance privacy.

The work of Massacci et al.~\cite{massacci2018futuresmex} proposes a distributed market exchange for futures assets that features a multi-step functionality, including a dark pool component. Their experiments show that the system can support up to ten traders. Notably, their model does not conceal orders; instead, it discloses an aggregated list of all pending buy and sell orders, which sets it apart from our solution. Moreover, there are existing works proposing private dark pool constructions utilizing blockchain technology~\cite{bag2019seal,galal2021publicly,ngo2021practical}, our focus diverges from this area. Furthermore, all these solutions experience slowdowns due to the reliance on computationally intensive public key cryptographic mechanisms.

The work of Hsu et al.~\cite{DBLP:journals/siamcomp/HsuHRRW16} also considered a private matching problem under the notion of \emph{joint differential privacy}, where the view of the adversary consists of the output received by all users except the user whose privacy is concerned.
Since the notion is still based on the conventional approach of using divergence on the adversarial views for neighboring inputs, their privacy notion can only lead to an almost optimal matching. In contrast, our new notion can achieve the exact optimal matching.

The authors in~\cite{PolychroniadouC24} employ differential privacy in a distinct and simplified setting of volume matching~\cite{BalchDP20,GamaCPSA22,polychroniadou2023prime}, where prices are predetermined and fixed, to obfuscate aggregated client volumes and conceal the trading activity of concentrated clients. In their auction mechanism, the obfuscated aggregate volumes are published daily, enabling buyers to make informed matching decisions based on this publicly available inventory. 
\section{Preliminaries}
\label{sec:prelim}

We first describe the problem setting and the adversary model.

\noindent \textbf{Problem Setup.} The order of each user~$i$ has a type $\tau_i \in {\mathsf{buy}, \mathsf{sell}}$, a price~$p_i \in \mathbb{R}+$, and a \textbf{positive} quantity~$x_i \in \mathbb{Z}_+$. An input configuration $C := {(\tau_i, p_i, x_i): i \in [n]}$ consists of users' orders. One unit of a buy order can be matched to one unit of a sell order if the buying price is at least the selling price. The goal is to design matching protocols that maximize the total number of matched units while preserving the privacy of users' data.

\noindent \textbf{Adversarial Model.}
The server can observe the most amount of information,
which is modeled as an adversary that is semi-honest, i.e.,
it will follow the protocol and try to learn about users' data.
Observe that when one unit of a buy order is matched
to a sell order, the identities of both the buyer and the seller,
as well as their bidding prices, must be revealed.
Here are possible privacy concepts.

\begin{compactitem}

\item Given the bid $(\tau_i, p_i, x_i)$
of a user~$i$, the type~$\tau_i$ and the price~$p_i$ are public information,
but the number~$x_i$ of units is private.
However, the identity of the user~$i$ is known, and because it is placing a bid,
it is also known that $x_i \geq 1$.

We will develop a new privacy notion that only hides the value of~$x_i$
when the bid cannot be fully executed. This theoretical privacy guarantee
will mainly focus on this new notion.

\item In practice, perhaps a user may not even want its presence in the system
to be known if no unit of its bid is executed. We will later describe
how this can be easily achieved by encrypting the user id, but this is
not the main technical focus.

\end{compactitem}

\ignore{
In the differential privacy framework, we specify the neighboring relation that describes what scenarios the adversary cannot tell apart.
}

\begin{definition}[Neighboring Input Configurations] 
Two input configurations $C$ and $\widehat{C}$ are neighboring if except for one user~$i$, all orders of other users are identical, and for user~$i$, only the quantity may differ by at most 1, i.e. $|x_i - \widehat{x}_i| \leq 1$. We denote $C \sim_i \widehat{C}$ in this case. 
\end{definition}

We review some basic probability tools that are commonly used in differential privacy~\cite{DBLP:conf/icalp/Dwork06}. A divergence is used to quantify how close two distributions on the same sample space are to each other.

\begin{definition} [Symmetric Hockey-Stick Divergence] 
Given distributions $P$ and $Q$ on the same sample space $\Omega$ and $\gamma \geq 0$, the symmetric hockey-stick divergence is defined as:

$$ \HS_\gamma(P \| Q) := \sup_{S \subseteq \Omega} \max\{ Q(S) - \gamma \cdot P(S), P(S) - \gamma \cdot Q(S)\}.$$

\end{definition}

\begin{remark} The hockey-stick divergence is related to the well-known $(\epsilon, \delta)$-differential privacy inequality. 

Specifically, $\HS_{e^\epsilon} (P | Q) \leq \delta$ if and only if for all subsets $S \subseteq \Omega$, 

$Q(S) \leq e^\epsilon \cdot P(S) + \delta$ and $P(S) \leq e^\epsilon \cdot Q(S) + \delta$. \end{remark}

\begin{definition}[Differential Privacy] Suppose when a matching protocol is run on a configuration $C$, the adversary can observe some information that is denoted by some random object $\msf{View}(C)$. Then, the protocol is $(\epsilon, \delta)$-differentially private if for all neighboring input configurations $C$ and $\widehat{C}$,

$\HS_{e^\epsilon}(\msf{View}(C) | \msf{View}(\widehat{C})) \leq \delta$. \end{definition}

Note that the maximum number of matched units can differ by 1 for neighboring input configurations. Therefore, to use the conventional notion of differential privacy, the matching protocol cannot guarantee that an optimal number of matched units is returned. In fact, it is not hard to see that to achieve $(\epsilon, \delta)$-DP, the protocol will need to match the number of units that is about $O(\frac{1}{\epsilon} \log \frac{1}{\delta})$ smaller than the maximum possile value.

\ignore{
This is consistent with the fact that the truncated geometric distribution can be used to achieve $(\epsilon, \delta)$-DP for two quantities that differ by at most 1.
}

\noindent \textbf{Truncated geometric distribution.} Let $Z$ 
be an even integer,
and $\alpha > 1$. The truncated geometric distribution $\Geom^Z(\alpha)$ has support with the integers in $[0..Z]$ such that its probability mass function at $x \in [0..Z]$ is proportional to $\alpha^{-| \frac{Z}{2} - x |}$. Specifically, the probability mass function at $x \in [0..Z]$ is 
\[
\frac{\alpha - 1}{\alpha + 1 - 2 \alpha^{-\frac{Z}{2}}} \cdot \alpha^{-\left|\frac{Z}{2} - x \right|}.
\]

\begin{fact}[Geometric Distribution and DP]
\label{fact:geom_dp}
For $\epsilon > 0$ and $0 \leq \delta < 1$, 
suppose $N$ is a random variable with distribution $\Geom^Z(e^\epsilon)$,
where $Z \geq \ceil{\frac{2}{\epsilon} \ln \frac{1}{\delta}}$ is even.
Then, for any integer~$n$,

$\HS_{e^\epsilon}(n+N \| n+1 + N) \leq \delta$.
\end{fact}

\subsection{Indifferential Privacy: A Relaxed Notion}\label{sec:IDP}

\noindent \textbf{High-Level Goal.} We would like to design a protocol that always returns the optimal number of matched units, but we will relax the conventional notion of differential privacy such that the bidding quantity~$x_i$ of a user~$i$ does not need privacy protection if the order is fully executed.

\noindent \textbf{Closest Refinement Pair.} To describe this new privacy notion, we will need some technical notation. The recent work~\cite{DBLP:journals/corr/abs-2406-17964} considers finite sample spaces which are sufficient for our purposes.

\begin{definition}[Refinement]
Suppose $\Omega_0$ and $\Omega_1$ are sample spaces
and $\mcal{F} \subseteq \Omega_0 \times \Omega_1$ is a binary relation,
which can also be interpreted as a bipartite graph $(\Omega_0 \cup \Omega_1, \mcal{F})$.
Given a distribution $\mathsf{P}_0$ on $\Omega_0$,
a refinement $\widehat{\mathsf{P}}_0$ of $\mathsf{P}_0$ (with respect to $\mcal{F}$)
is a distribution on $\mcal{F}$ such that
for every $i \in \Omega_0$, $\mathsf{P}_0(i) = \sum_{j: (i,j) \in \mcal{F}} \widehat{\mathsf{P}}_0(i,j)$.

A refinement for a distribution on $\Omega_1$ is a distribution on $\mcal{F}$ defined analogously.
\end{definition}

\begin{figure}[hbt!]
    \centering
    \includegraphics[width=\columnwidth]{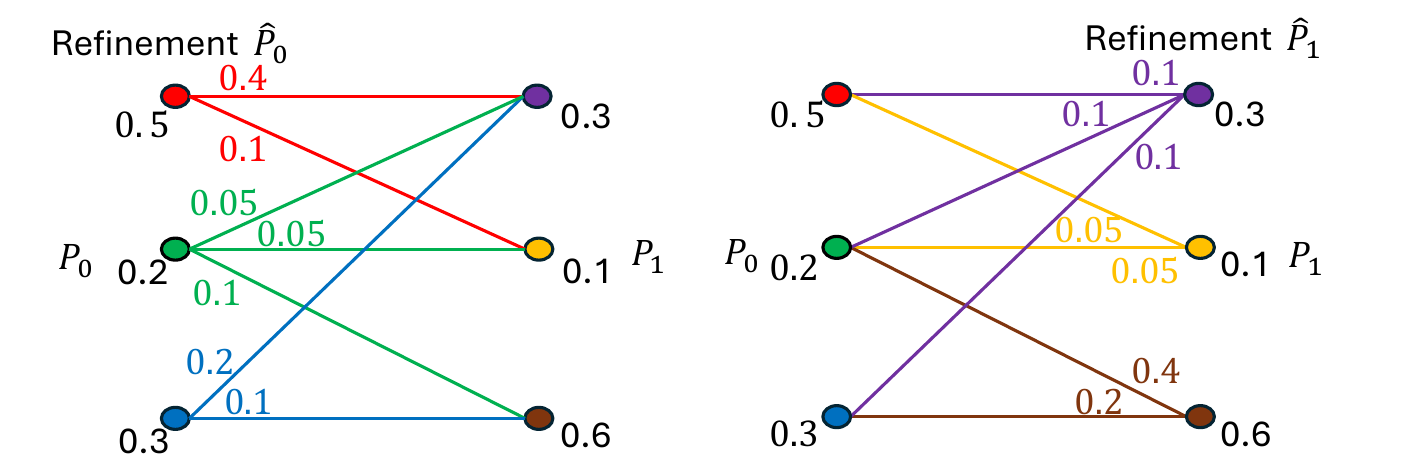}  
    \caption{\textbf{Refinement Pair.} \textnormal{The figure
		shows an example of a bipartite graph $(\Omega_0, \Omega_1; \mcal{F})$, where each node has a probability mass.  Each subfigure shows
		the probability distribution refinement on each side.}}
    \label{fig:refinement}
\end{figure}

\begin{definition}[Closest Refinement Pair]
Given distributions $\mathsf{P}_0$ and $\mathsf{P}_1$
on sample spaces that are equipped with some relation $\mcal{F}$
and a divergence notion $\mathsf{D}$,
the divergence between $\mathsf{P}_0$ and $\mathsf{P}_1$ with respect to $\mcal{F}$
is defined to be:

$$\mathsf{D}^{\mcal{F}}(\mathsf{P}_0 \| \mathsf{P}_1) =
\inf_{(\widehat{\mathsf{P}}_0, \widehat{\mathsf{P}}_1)} \mathsf{D}(\widehat{\mathsf{P}}_0 \| \widehat{\mathsf{P}}_1),$$

where the infimum is taken over all refinements 
$\widehat{\mathsf{P}}_0$ and $\widehat{\mathsf{P}}_1$
of $\mathsf{P}_0$ and $\mathsf{P}_1$, respectively.
\end{definition}

\begin{remark}
As we shall see, in our dark pool application,
the relation $\mcal{F}$ in $\Omega_0 \times \Omega_1$,
is actually much simpler.
In particular, each node from one side
has at most one neighbor on the other side.
Hence, for the distribution on each side, there is only one possible
refinement.
Even though the concept of closest refinement pair is not needed for this application,
we still give a more general Definition~\ref{defn:idp} which may be relevant
in future applications.
\end{remark}

\begin{fact}[Universal Closest Refinement Pair~\cite{DBLP:journals/corr/abs-2406-17964}]
Given the above distributions 
 $\mathsf{P}_0$ and $\mathsf{P}_1$, together with the relation $\mcal{F}$ on
their corresponding sample spaces,
there exists a universal refinement pair
$(\widehat{\mathsf{P}}_0, \widehat{\mathsf{P}}_1)$ that minimizes
$\mathsf{D}(\widehat{\mathsf{P}}_0 \| \widehat{\mathsf{P}}_1)$ for all
divergences $\mathsf{D}$ satisfying the data-processing 
inequality\footnote{
A divergence $\mathsf{D}$ satisfies
the data processing inequality if for 
any pair of joint distributions $(X_0, Y_0)$ and $(X_1, Y_1)$,
the corresponding marginal distributions satisfy
$\mathsf{D}(X_0 \| X_1)
\leq \mathsf{D}((X_0, Y_0) \| (X_1, Y_1))$.
} 
(which includes the hockey-stick divergence $\HS$).
\end{fact}

\noindent \textbf{Intuition.} To fit the above notation to
our problem, let us consider a pair of neighboring input configurations
$C_0$ and $C_1$, in which the quantities of some user~$i$ differ by 1.
For instance, in $C_0$, the user~$i$ has $n_0 > 0$ units, while in $C_1$,
the user has $n_1 = n_0 + 1$ units.

For those two input configurations,
the corresponding sample spaces of the adversary's views
are $\Omega_0$ and $\Omega_1$, which may not be the same,
because the supports of the views for the two configurations may be different.

Next, let us look at each point in the sample space more carefully.
As the protocol is being executed, the adversary gains more information step-by-step.
Hence, each point~$\sigma$ in the sample space is
a time-series sequence $\sigma = (\sigma_1, \sigma_2, \ldots)$.

\begin{definition}[Indifference Relation]
\label{defn:indifference}
Given the sample spaces $\Omega_0$ and $\Omega_1$ of
the adversary's views corresponding
to neighboring input configurations~$C_0$ and $C_1$ as above,
we define an indifference relation $\mcal{F} \subseteq \Omega_0 \times \Omega_1$
as follows, where a pair of time sequences $(\sigma^{(0)}, \sigma^{(1)}) \in \mcal{F}$ is related \emph{iff}

\begin{compactitem}

\item  The two sequences $\sigma^{(0)} = \sigma^{(1)}$ are the same; or

\item The two time sequences $\sigma^{(0)}$ and $\sigma^{(1)}$
have a maximal common prefix $\rho$ such that $\rho$ implies that
$n_0$ units of user~$i$'s order have been matched, i.e.,
in $C_0$, the order of user~$i$ is fully executed.
\end{compactitem}

\end{definition}

With the mathematical concepts formalized,
we are ready to introduce our new notion of privacy.

\begin{definition}[Indifferential Privacy (IDP)]
\label{defn:idp}
Suppose when a matching protocol is run on a configuration $C$,
the adversary can observe some information that is denoted by 
some random object $\msf{View}(C)$.
Then, the protocol is $(\epsilon, \delta)$-indifferentially private (IDP)
if for all neighboring input configurations $C$ and $\widehat{C}$
with the appropriate indifference relation~$\mcal{F}$ defined,

$\HS^{\mcal{F}}_{e^\epsilon}(\msf{View}(C) \| \msf{View}(\widehat{C})) \leq \delta$.
\end{definition}

\noindent \textbf{Intuition.}
In Definition~\ref{defn:indifference},
if we only have the first bullet $\sigma^{(0)} = \sigma^{(1)}$,
then this reduces to the usual differential privacy.
The second bullet says the indifference relation specifies when 
the user~$i$ is indifferent
about whether the adversary can distinguish between the 
two views in a pair $(\sigma^{(0)}, \sigma^{(1)}) \in \mcal{F}$.

In our application scenario, this captures the idea
that when a user's order is fully executed, then the privacy of its bidding quantity
no longer needs protection.

\noindent \textbf{Information Theoretic vs Computational IDP.}
When we describe our protocol,
we will use the ideal functionality of cryptographic
primitives and prove the privacy under
Definition~\ref{defn:idp} that is information theoretic.
If we replace the ideal functionality with
the real-world cryptographic construct,
we can achieve the following notion of
computational IDP that is analogous to SIM-CDP 
introduced in Mironov et al.~\cite{DBLP:conf/crypto/MironovPRV09}.

\begin{definition}[Computational IDP (CIDP)]
\label{defn:cidp}
A protocol $\Pi$ is $(\epsilon, \delta)$-CIDP
if there exists a protocol $\Pi'$ that
is $(\epsilon, \delta)$-IDP such that
$\Pi$ and $\Pi'$ are computationally indistinguishable.
\end{definition}

\noindent \textbf{Technical Focus.}
In this work, we focus on showing
our protocol using ideal functionality of cryptographic primitives satisfies Definition~\ref{defn:idp}.
It is straightforward to apply standard hybrid arguments~\cite{DBLP:conf/crypto/MironovPRV09}
to replace an ideal functionality with the real-world cryptographic
primitive to show that the resulting protocol
satisfies Definition~\ref{defn:cidp}.

\ignore{

\begin{definition}[Multiplicative Neighboring Relation]
For $\rho \geq 0$, two positive numbers $x$ and $x'$ are $\rho$-multiplicative neighboring, if
$| \ln \frac{x}{x'}| \leq \rho$.

\end{definition}

\begin{definition}[Truncated Laplace Distribution]
The truncated Laplace distribution $\mathsf{Lap}(\mu, b, C)$
has a support in $[\mu - C, \mu + C]$
and for $t \in [\mu - C, \mu + C]$,
its probability density function is proportional to: $\exp(- \frac{|t - \mu|}{b})$.
\end{definition}

\begin{fact}
Let $\epsilon \geq 0$ and $0 < \delta < 1$.
Suppose $x$ and $x'$ are $\rho$-multiplicative neighboring,
and $Z$ is a random variable have truncated Laplace distribution
$\mathsf{Lap}(\mu, b, C)$,
where $b = \frac{\rho}{\epsilon}$ and $\mu = C = b \ln \frac{1}{\delta}$.

Then, the two distributions $x \cdot \exp(Z)$ and $x' \cdot \exp(Z)$ are $(\epsilon, \delta)$-close.

$\exp(Z) \in [1, (\frac{1}{\delta})^{ \frac{2 \rho}{\epsilon} }]$

Social welfare suffers by a factor 
of $(\frac{1}{\delta})^{ \frac{4 \rho}{\epsilon}}$
\end{fact}

}

\ignore{

\begin{definition}[Conditional Differential Privacy]
A mechanism $\mathsf{M}$
is $(\epsilon, \delta)$-conditional differentially private if
for all neighboring input configurations $C \sim_i \widehat{C}$ such that
some specified condition $\mathcal{U}_i$ 
(that may depend on~$i$) has non-zero probability in both scenarios,
we have for any subset $S$ of outputs,

$\Pr[\mathsf{M}(C) \in S \, | \mathcal{U}_i]
\leq e^\epsilon \cdot \Pr[\mathsf{M}(\widehat{C}) \in S \, | \mathcal{U}_i] + \delta$.

\end{definition}
}

\section{Indifferentially Private Matching Protocol}
\label{sec:algo}

We describe the components of our private matching protocol.

\noindent \textbf{Bid Type.}
For a user~$i$ with $\mathsf{Id}_i$, its bid has a type~$\tau_i \in \{\mathsf{Buy}, \mathsf{Sell}\}$ with price $p_i$
and quantity $x_i \in \Z_{\geq 0}$.

\noindent \textbf{Bid Submission.}
Given the bid $(\mathsf{Id}_i, \tau_i, p_i, x_i)$
from user~$i$, 
suppose a user wants to achieve $(\epsilon, \delta)$-IDP.
Then,
it performs the following during bid submission.
\begin{enumerate}

\item Sample non-negative noise $N_i \in \Z$ 
according to the truncated geometric distribution $\Geom^Z(e^\epsilon)$,
for some even $Z \geq \ceil{\frac{2}{\epsilon} \ln \frac{1}{\delta}}$.

\item Create $y_i := x_i + N_i$ nodes of type $(\tau_i, p_i)$,
where the type is known to the server.

Out of the $y_i$ nodes, $x_i$ are real nodes and $N_i$ are fake nodes.
Each node contains encrypted information of whether it is real or fake.

The list of created nodes are sent to the server, where all real nodes
appear before the fake nodes.  
Observe that the server cannot distinguish between
the real and the fake nodes.
\end{enumerate}

\noindent \textbf{Cryptographic Assumptions.}  Depending on what security guarantee of the final protocol is needed, 
we can assume different properties of the ciphertext.

\begin{itemize}

\item To achieve just our notion of indifferential privacy,
the encrypted information of a node can be achieved by 
a non-malleable commitment scheme~\cite{DBLP:journals/jacm/LinP15}.
In practice, we can also use a hash function such as AES.

\textbf{Ideal Commitment Scheme.}
For ease of exposition, we describe our protocol using
an ideal commitment scheme.  A user creates a commitment of a message,
which is an opaque object whose contents the adversary cannot observe.
At a later time, the user can choose to open the commitment, which
must return the original message. In particular, a cryptographic commitment scheme has these key properties: (1) Hiding: The value stays secret until revealed.
(2) Binding: The committed value can't be changed.
(3) Unforgeability: No one can forge or tamper with the commitment.
(4) Correctness: The committed value is always correctly revealed.

\item Assuming the shuffler model, the server does not know
which nodes come from which user.  Hence, 
in this case, each node also contains encrypted information about the user ID.
However, the server still knows which nodes originate from the same user,
and the nodes from the same user are sorted in a list such that real nodes appear 
first before the fake nodes.
\end{itemize}

\noindent \textbf{Node Popularity.}
In any case, the server can create an ordering on the nodes on the buyer and the seller sides
according to the price.
For a buy node, a higher price is more popular; for a sell node, a lower price is more popular. 
On each side, the most or the least popular price is known as an \emph{extreme} price.

\noindent \textbf{Matching Graph.}
The algorithm maintains a bipartite graph
 $G = (\mcal{B}, \mcal{S}; E)$,
where there is an edge between a buy and a sell node
if the prices are compatible, i.e., the buying price
is at least the selling price.
As some nodes are matched during the process,
we assume that the algorithm will immediately remove any isolated node
that no longer has a neighbor.

\noindent \textbf{Maximum Matching Problem.}
Given a bipartite graph $G = (\mcal{B}, \mcal{S}; E)$,
a matching $M \subseteq E$ is a subset of edges
such that no two edges in $M$ are incident on the same node.
The \emph{maximum matching problem} aims to find a matching
$M$ with the maximum number of edges.  We next describe a
strategy to find a maximum matching on a bipartite graph
induced by buyer and seller nodes.

\noindent \textbf{Polar Opposite.}
Suppose the current matching graph $G$ has no isolated nodes.
For an arbitrary side, 
observe that any node on this side
with the most popular price has an edge connected to
any node on the other side 
with the least popular price.
We say that two such nodes from the two sides are \emph{polar opposite} of each other.
The following fact shows that an optimal matching can be returned by iteratively matching
polar opposites.

\begin{fact}
\label{fact:basic_matching}
Suppose a matching graph $G$ has no fake or isolated nodes,
and $u$ and $v$ are any two nodes from the different sides that
are polar opposite of each other.
Then, there exists a maximum matching in $G$ in which $u$ and $v$ are matched.
\end{fact}

\begin{proof}
Without loss of generality, suppose $u$ has the most popular price
on its side, while $v$ has the least popular price on its side.
Suppose $M$ is a maximum matching in which $u$ and $v$ are not matched
to each other.  We will modify $M$ without decreasing the matching size such that $u$ and $v$
are matched to each other in the following steps.

\begin{enumerate}

\item If $u$ is not matched to any node, we can make $u$ replace
any other node on its side, because it has the most popular price.

Hence, we may assume that $u$ is matched in $M$.

\item If $v$ is not matched in $M$, then we can replace the partner of~$u$
with $v$ to make $u$ and $v$ matched to each other.

\item Otherwise, we have $u$ is matched to some $v'$ and $v$ is matched to some $u'$.
Observe that if $u'$ is compatible with the least popular price on the other side,
then it must be compatible with $v'$.  Hence, $u$ and $u'$ can swap partners
such that $u$ and $v$ become matched.
\end{enumerate}
\end{proof}


\noindent \textbf{Matching Procedure with Fake Nodes.}
With our assumption, for each node type,
we assume that the nodes are sorted such that
the nodes from the same user are together and its real nodes appear first.

In each step, the server picks nodes $u$ and $v$ from the two sides
that are polar opposite of each other.
The owners of nodes $u$ and $v$ will participate in the following protocol:

\begin{enumerate}

\item If the node of an owner is real,
the owner will open its commitment to reveal the node is real.

\item If the node of an owner is fake,
this implies that all real nodes of that owner are already
matched and the owner will reveal all its fake nodes.

\item If both the nodes $u$ and $v$ are opened to be real,
the nodes $u$ and $v$ are mathced and removed from the matching graph.

\item If one of the nodes is real and the other node is fake,
the real node will be considered in the next iteration if it still has
a polar opposite in the remaining graph.

\end{enumerate}

\noindent \textbf{Correctness Intuition.}  The protocol
is outlined in Algorithm~\ref{alg:matching};
note that any arbitrary unspecified choice made by the algorithm
can either be randomized or deterministic according to some additional rules.

\begin{lemma}
Algorithm~\ref{alg:matching} returns
a maximum matching between real buy and sell nodes in the matching graph.
\end{lemma}

\begin{proof}
Comparing with the case with no fake nodes,
observe that whenever a fake node is encountered,
it will be removed immediately.  Hence, the matching behavior
is exactly the same as the case with no fake nodes.
From Fact~\ref{fact:basic_matching},
a maximum matching is returned.
\end{proof}

\begin{algorithm}[H]
\caption{Matching Protocol with Fake Nodes} \label{alg:matching}
\SetAlgoLined

Matching graph $G \gets (\mcal{B}, \mcal{S}; E)$

$u \gets \msf{NULL}$

$M \gets \emptyset$

\While{$E \neq \emptyset$}{

    \If{$u = \msf{NULL}$}{
			From any side and any extreme price on that side, select the next node~$u$.
		}
		
		Select the next node~$v$ that is a polar opposite to~$u$.

		The algorithm announces that an attempt is made to match $u$ and $v$.
		
		The owners of the two nodes open their commitments (if they have not already done so) and reveal
		whether the nodes are real or fake.

		The algorithm infers that if an owner reveals a fake node,
		all the remaining nodes by the same owner are also fake.

    \eIf{either $u$ or $v$ is fake}{
        			
        Any fake node and its incident edges are removed from $G$.
    }{
        \If{both $u$ and $v$ are real}{
            					
            The pair $(u,v)$ is matched and added to $M$;
						remove $u$ and $v$ from $G$.
						
        }
    }
    Remove any isolated node from $G$.
		
		If exactly one of $u$ and $v$ is still unmatched and remains
		in the graph, set $u$ to be that node; otherwise,
		set $u$ to $\msf{NULL}$.
				
}

	\textbf{return} matching $M$
\end{algorithm}

\subsection{Privacy Analysis}

We show that our dark pool auction protocol satisfies IDP.

\ignore{

We first recall a result the generalizes Fact~\ref{fact:geom_dp}.
\begin{fact}
Suppose $P$ and $Q$ are two distributions on a finite sample space~$\Omega$
such that the following conditions hold.

\begin{compactitem}

\item  Denote $S_P := \{\omega \in \Omega: P(\omega) > 0 \wedge Q(\omega) = 0\}$
and $S_Q := \{\omega \in \Omega: P(\omega) = 0 \wedge Q(\omega) > 0\}$.
We have $\max \{P(S_P), Q(S_Q)\} \leq \delta$.

\item For any $\omega \in \Omega$ such that $P(\omega) \cdot Q(\omega) > 0$,
we have
$e^{-\epsilon} \leq \frac{P(\omega)}{Q(\omega)} \leq e^{\epsilon}.$

\end{compactitem}

Then, we have $\HS_{e^\epsilon}(P \| Q) \leq \delta$.

\end{fact}

}

\noindent \textbf{Neighboring Input Configurations.}
Recall that we consider neighboring input configurations
$C_0$ and $C_1$ in which exactly one user~$i$ has 
different bidding quantities.
Suppose in $C_0$, the quantity is $n_0$ and
in $C_1$, the quantity is $n_1 = n_0 + 1$.

\noindent \textbf{Adversarial View.} We decompose
the view space into $\Gamma \times \Lambda$,
where $\Gamma$ corresponds to 
the number of nodes created by user~$i$,
and $\Lambda$ includes (i) the bid submissions by all other users,
and (ii) any random bits
used to make choices in the matching procedure in Algorithm~\ref{alg:matching}.
The distributions on $\Gamma$ differ in $C_0$ and $C_1$,
but in the two scenarios, the distributions on $\Lambda$ are identical
and independent of the $\Gamma$ component.

Note that any other information observed by the adversary can be recovered
by a point in $\Gamma \times \Lambda$.

\noindent \textbf{Indifference Relation.}  
We define an almost trivial relation on $\Gamma \times \Lambda$, namely,
$(\gamma_0, \lambda_0) \sim (\gamma_1, \lambda_1)$ \emph{iff}
$\gamma_0 = \gamma_1$ and $\lambda_0 = \lambda_1$.
This means that the same number of
nodes are submitted by user~$i$ in the two scenarios;
moreover, all other users submit exactly the same bids in the two scenarios,
and the also the same random bits are used in the matching process.

However, this does not mean that the views of the adversary are the same in
both scenarios, because the number of real nodes for user~$i$ are different
in $C_0$ and $C_1$.  Nevertheless, if the views are different,
it must be because $n_0$ real nodes of user~$i$ have been matched
in both scenarios, and the algorithm tries to match the next node from user~$i$.
This is consistent with the indifference relation
described in Definition~\ref{defn:indifference}.

\begin{lemma}
Using an ideal commitment scheme,
the dark pool auction protocol is $(\epsilon, \delta)$-IDP.
\end{lemma}

\begin{proof}
Recall that from Definition~\ref{defn:idp},
our goal is to prove an upper bound for:
$\HS^{\mcal{F}}_{\epsilon}(\msf{View}(C_0) \| \msf{View}(C_1))$.

Because of the ideal commitment scheme,
during the bid submission, the server can only see
how many nodes are created by each user.
From user~$i$, the information collected at this stage 
is a point in $\Gamma$.
Suppose that in $C_0$, the distribution of the number of created nodes by user~$i$ is $P_0$;
and in $C_1$, the corresponding distribution is $P_1$.
Note that $P_0$ and $P_1$ are both distributions on $\Gamma$.
From Fact~\ref{fact:geom_dp},
we have: 
$\HS_{e^\epsilon}(P_0\| P_1) \leq \delta$.

According to the above discussion,
for both $C_0$ and $C_1$,
we have the same distribution $Q$ on $\Lambda$
that represents other users' bids and the random bits used by the server during
the matching process.

As argued above, a point in $\Gamma \times \Lambda$
is sufficient to recover the view of the adversary.
Hence,
we have:

$\HS^{\mcal{F}}_{e^\epsilon}(\msf{View}(C_0) \| \msf{View}(C_1))
= \HS_{e^\epsilon}((P_0, Q) \| (P_1, Q))$.

Moreover, observe that if $(\gamma_0, \lambda_0) = (\gamma_1, \lambda_1)$
correspond to two points in $\Gamma \times \Lambda$
under configurations $C_0$ and $C_1$, respectively,
it must be the case that during bid submission,
exactly $\gamma_0 = \gamma_1$ nodes are created by user~$i$.
Furthermore, everything is identical in both scenarios,
except that the $(n_0 + 1)$-st node in user~$i$'s list
is fake in $C_0$ and is real in $C_1$.
Therefore, the adversary either has identical views
in both scenarios (because at most
$n_0$ nodes in user~$i$'s list have been attempted to be matched), or else it must be the case that 
all real nodes of user~$i$ in $C_0$ have been fully matched;
this is consistent with the indifference relation~$\mcal{F}$.

However, note that $Q$ is independent of the $\Gamma$ component.
In general, the hockey-stick divergence (as well as other commonly used divergences)
satisfies the property that observing the same extra independent randomness
should not change the value of the divergence.
Hence, 
we have:
$\HS_{e^\epsilon}((P_0, Q) \| (P_1, Q)) = \HS_{e^\epsilon}(P_0\| P_1)$,
which we have already shown is at most~$\delta$.  This completes the proof.
\end{proof}
\ignore{
Using standard hybrid arguments~\cite{DBLP:conf/crypto/MironovPRV09}, we get the following corollary immediately.
}

\begin{corollary}
Replacing the ideal commitment scheme with a real-world commitment
scheme (such as~\cite{DBLP:journals/jacm/LinP15}),
the dark pool auction protocol is $(\epsilon, \delta)$-computational IDP.
\end{corollary}

\section{Implementation and Evaluation}

\subsection{Implementation}
We provide a comprehensive end-to-end implementation of our indifferentially private dark pool auction. We have also implemented a non-private auction to compare the total computation and communication time. Our code is available at \url{https://github.com/adyaagrawal/idp-darkpool}. 

Both the non-private and Indifferentially Private (IDP) protocols were incorporated into ABIDES~\cite{ByrdHB20},\footnote{ABIDES has also been used in simulating privacy preserving federated learning protocols such as the most recent work of~\cite{MaWAPR23}.} an open-source high-fidelity simulator tailored for AI research in financial markets such as stock exchanges. ABIDES is ideal for this purpose, as it supports simulations with tens of thousands of clients interacting with a central server for transaction processing, along with customizable pairwise network latencies to simulate real-world communication delays. We run the simulations on a personal x64-based Windows PC equipped with a single Intel Core i5-10210U processor running at 1.6 GHz and 16 GB of DDR4 memory. 

ABIDES employs the cubic network delay model, where the latency is determined by a base delay (within a specified range) and a jitter that influences the percentage of messages arriving within a designated timeframe, thereby shaping the tail of the delay distribution. Our experiments were conducted using three distinct network settings: local (client machines in NYC), global (client machines from NYC to Sydney), and very wide network (client machines across the world). For each setting, we configured the base delay according to Table~\ref{tab:network_delays} while employing the ABIDES's default parameters for jitter.

\begin{table}[h]
    \centering
    \begin{tabular}{|c|c|}
        \hline
        \textbf{Network Setting} & \textbf{Base Delay (ms)} \\
        \hline
        Local (Within NYC)              & 0.021 - 0.1                           \\
        Global (NYC to Sydney)   & 21 - 53    \\
        Very Wide Network (Across the World)        & 10 - 100                                \\
        \hline
    \end{tabular}
    \caption{Base delay for different network configurations.}
    \label{tab:network_delays}
\end{table}

In our non-private implementation, clients submit their orders without concealing their identities. These orders are organized into two doubly linked lists: one for buy orders and another for sell orders. Our implementation follows the maximum matching algorithm explained in Section~\ref{sec:algo}. Once matches are identified, they are sent back to the clients for execution.

For our IDP implementation, clients generate an array of sorted orders, with a few randomly numbered fake ones to ensure indifferential privacy.
They commit to their identity and whether the order is real or fake using the Cryptodome library in Python based on AES. The maximum matching algorithm is employed and if the match contains the order that the individual client has placed, the client opens the commitment to the server to reveal whether the order is real or fake. Then, the server requests the client to reveal the identity or tries to find a new match. Once both the parties have opened their identity and revealed it to the server, the orders are executed and the next round of matching starts.

\begin{figure*}[t]
    \centering
    \begin{subfigure}[t]{0.3\textwidth}  
        \centering
        \includegraphics[width=\textwidth]{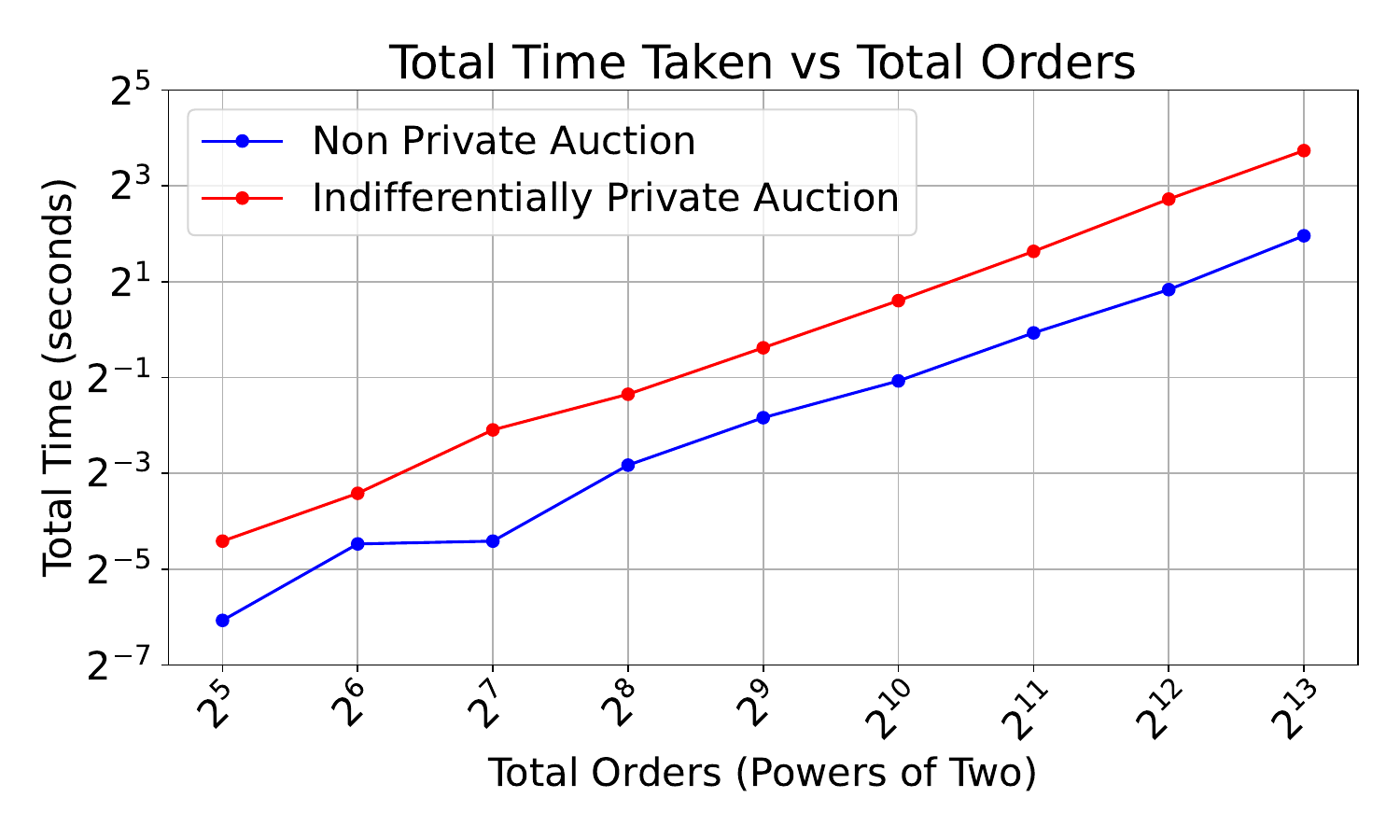}  
        \caption{\textnormal{Local setting.}}
        \label{fig:local}
    \end{subfigure}
    \hspace{1em} 
    \begin{subfigure}[t]{0.3\textwidth}  
        \centering
        \includegraphics[width=\textwidth]{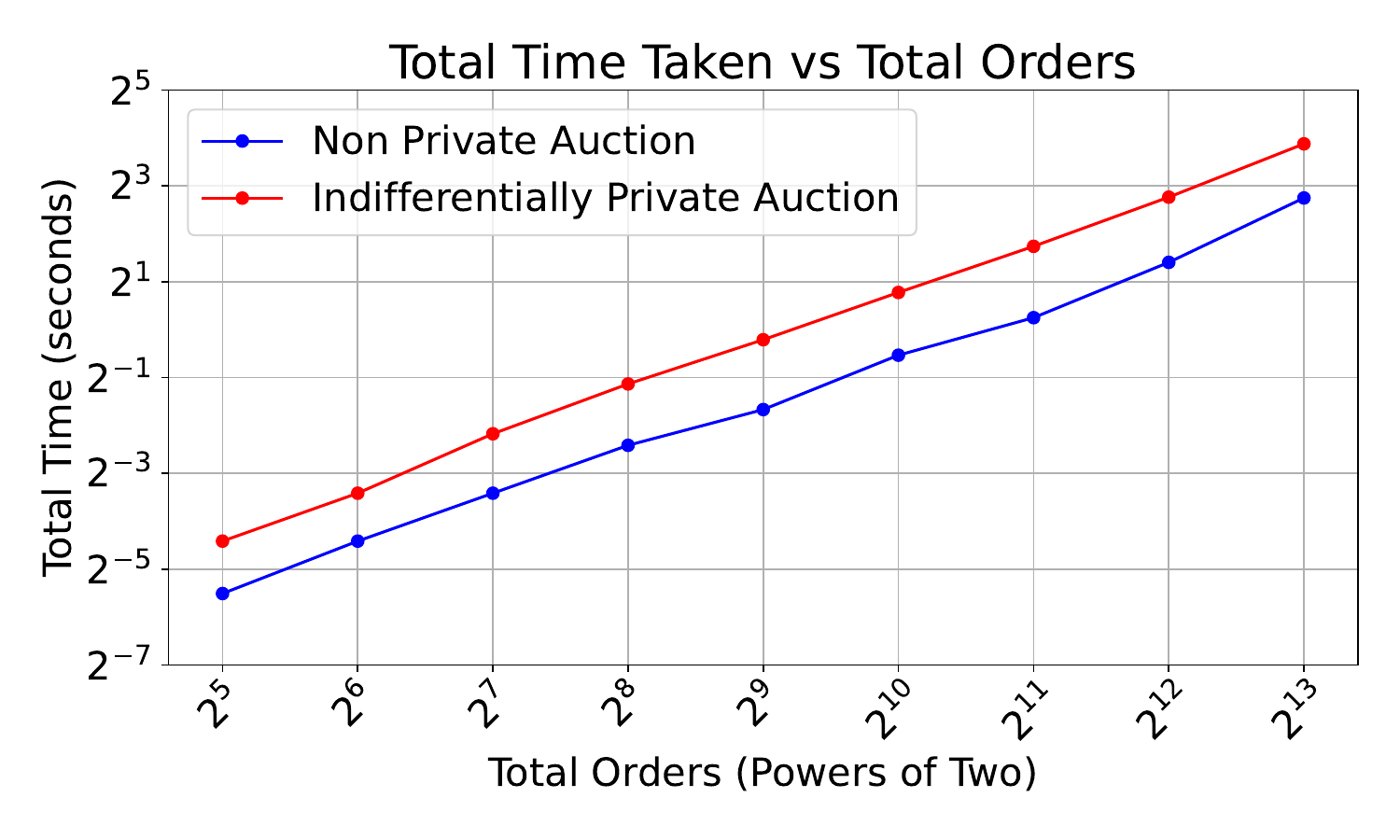}  
        \caption{\textnormal{Global setting.}}
        \label{fig:global}
    \end{subfigure}
    \hspace{1em} 
    \begin{subfigure}[t]{0.3\textwidth}  
        \centering
        \includegraphics[width=\textwidth]{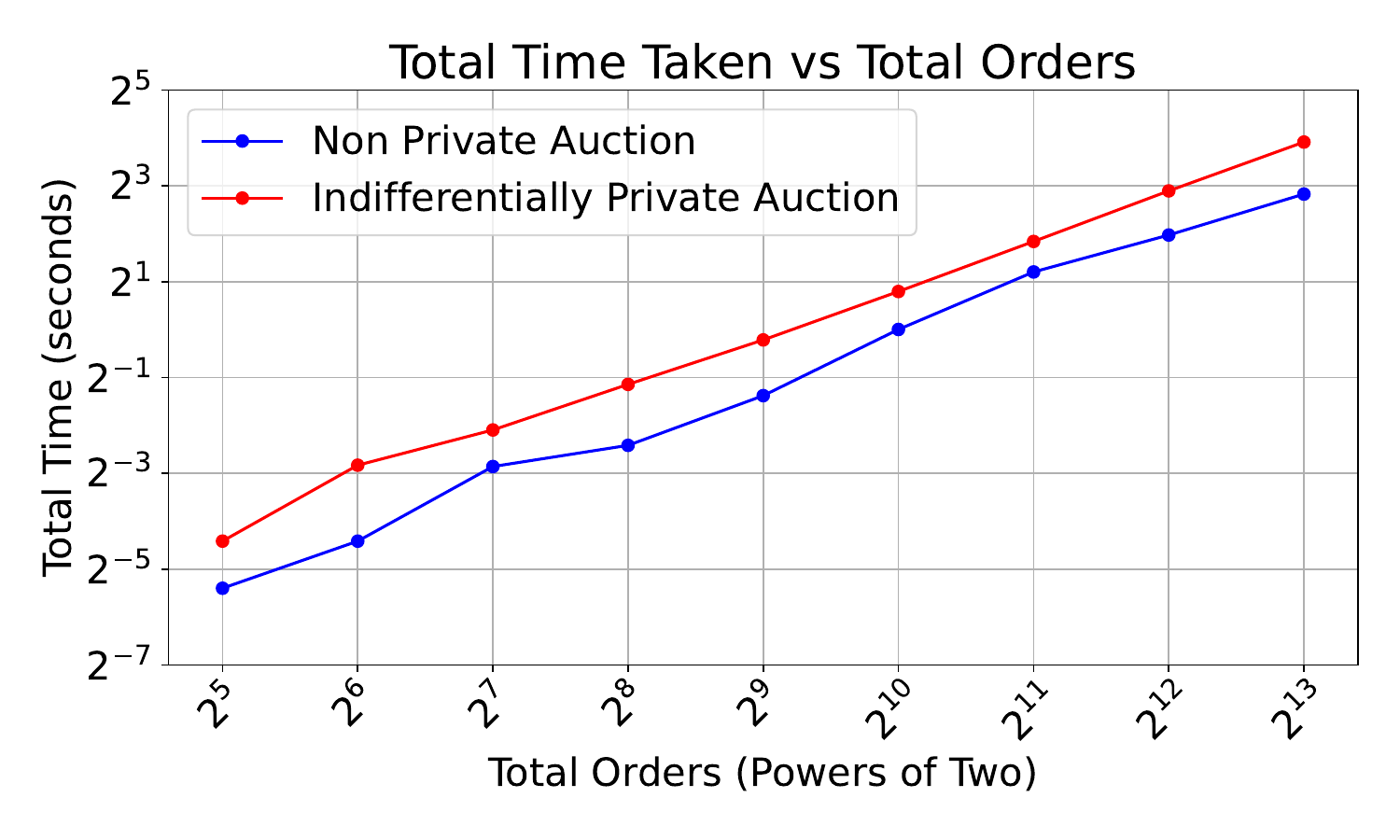}  
        \caption{\textnormal{Very Wide Network setting.}}
        \label{fig:world}
    \end{subfigure}
    \caption{\textbf{Comparison of Time Taken Across Different Settings.}}
    \label{fig:comparison}
\end{figure*}

\subsection{Experimentation}

\noindent \textbf{Performance comparison with non-private implementation.}
To evaluate the performance of our order matching system, we compared the total time required to match orders using both the non-private and Indifferentially Private (IDP) implementations in all three different network settings. We conducted experiments with a range of parties each generating \(2^3\) orders, starting from 4 parties (\(2^5\) orders) and extending up to 1024 parties (\(2^{13}\) orders).

The data generation process is designed to ensure that each set of randomly generated orders consistently achieves a matching rate of 70\% to 80\%.\footnote{When the \% of matches is lower, the total running time naturally decreases, as fewer matches result in less overhead.} For instance, each client generates 8 sorted orders, with 5 to 7 being real and 1 to 3 being fake. The prices for buy orders are randomly selected between 99 and 101, while sell orders are priced between 98 and 100.

\ignore{
\begin{figure}[htbp]
    \centering
    \begin{subfigure}[t]{0.5\textwidth}  
        \centering
        \includegraphics[width=\textwidth]{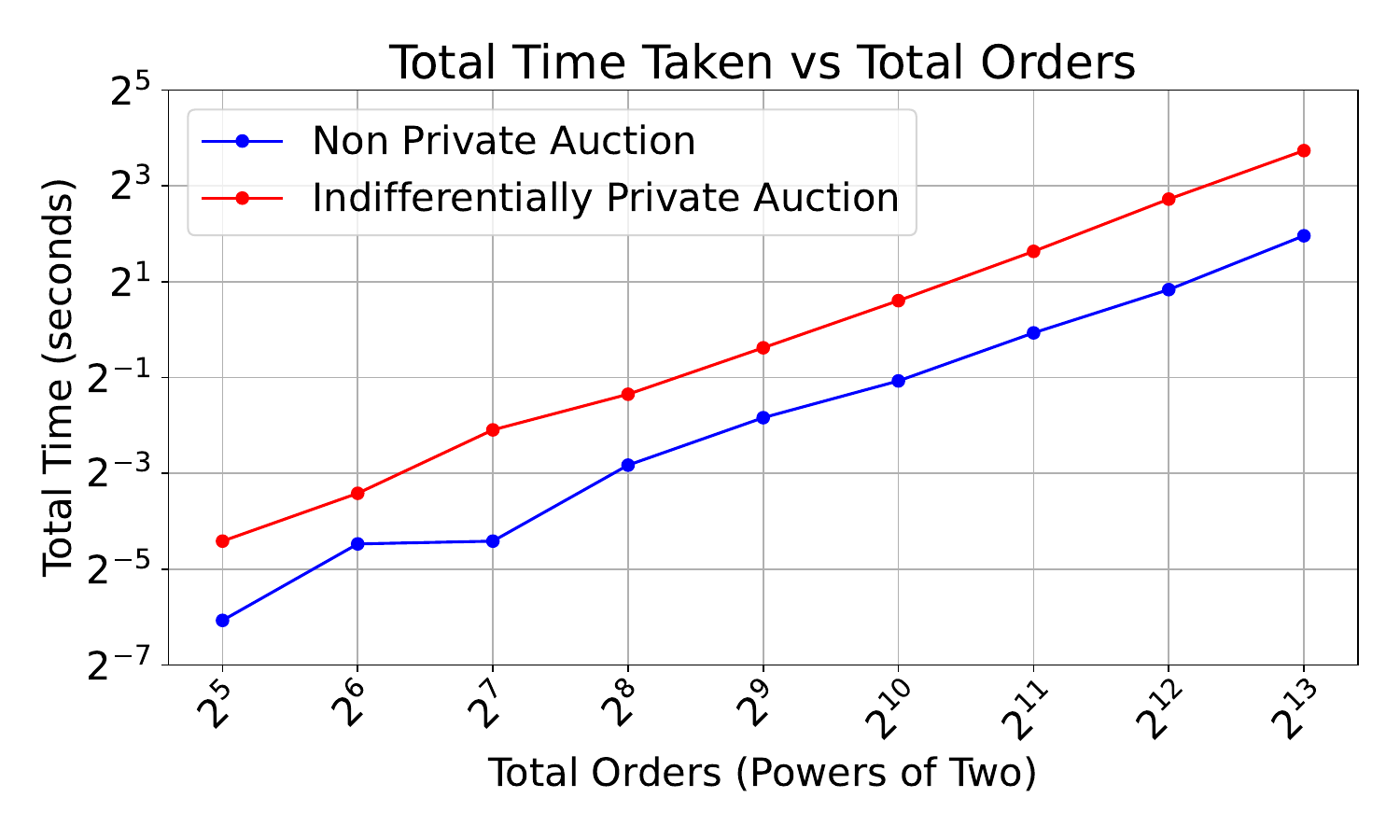}  
        \caption{\textnormal{Local setting.}}
        \label{fig:local}
    \end{subfigure}
    \vspace{1em} 
    \begin{subfigure}[t]{0.5\textwidth}  
        \centering
        \includegraphics[width=\textwidth]{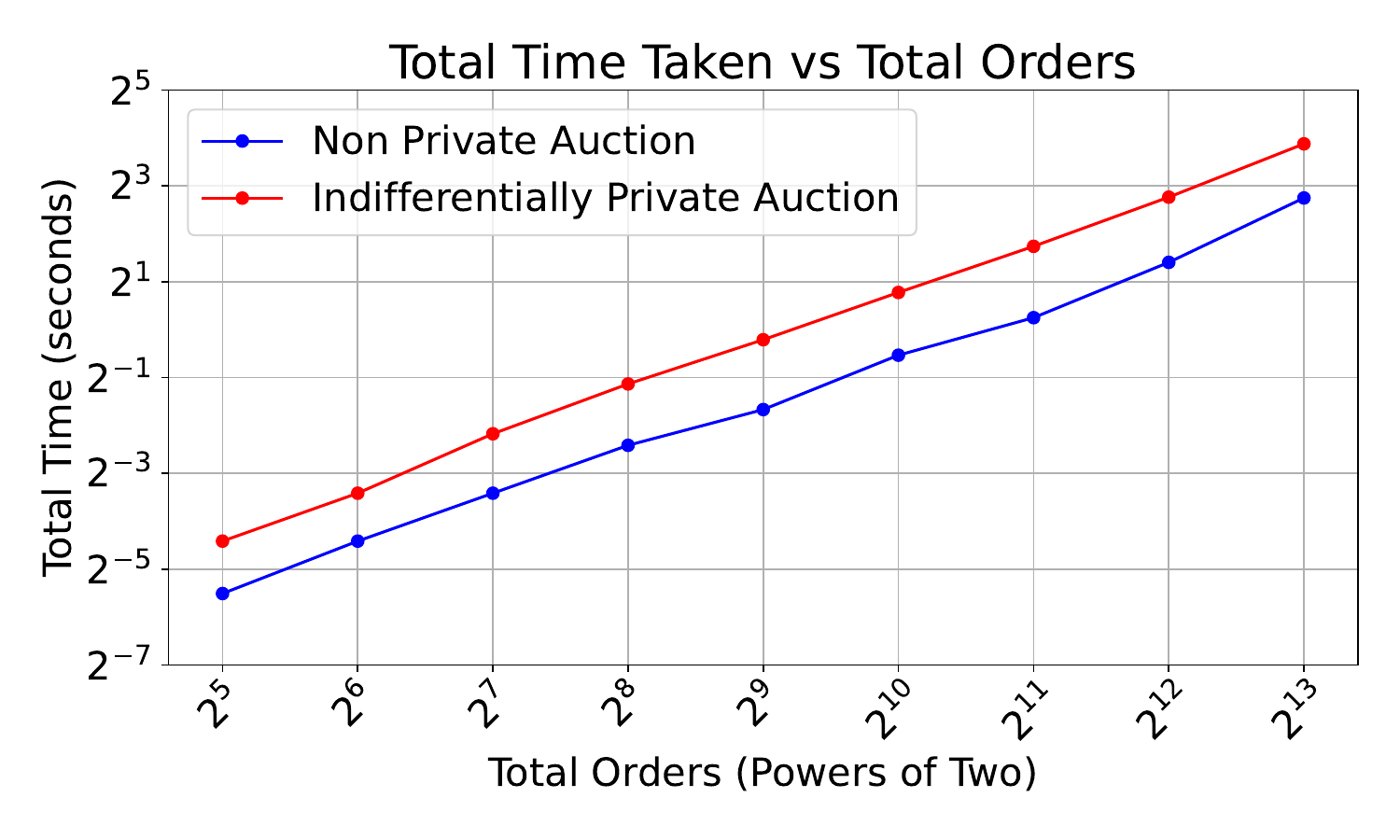}  
        \caption{\textnormal{Global setting.}}
        \label{fig:global}
    \end{subfigure}
    \vspace{1em} 
    \begin{subfigure}[t]{0.5\textwidth}  
        \centering
        \includegraphics[width=\textwidth]{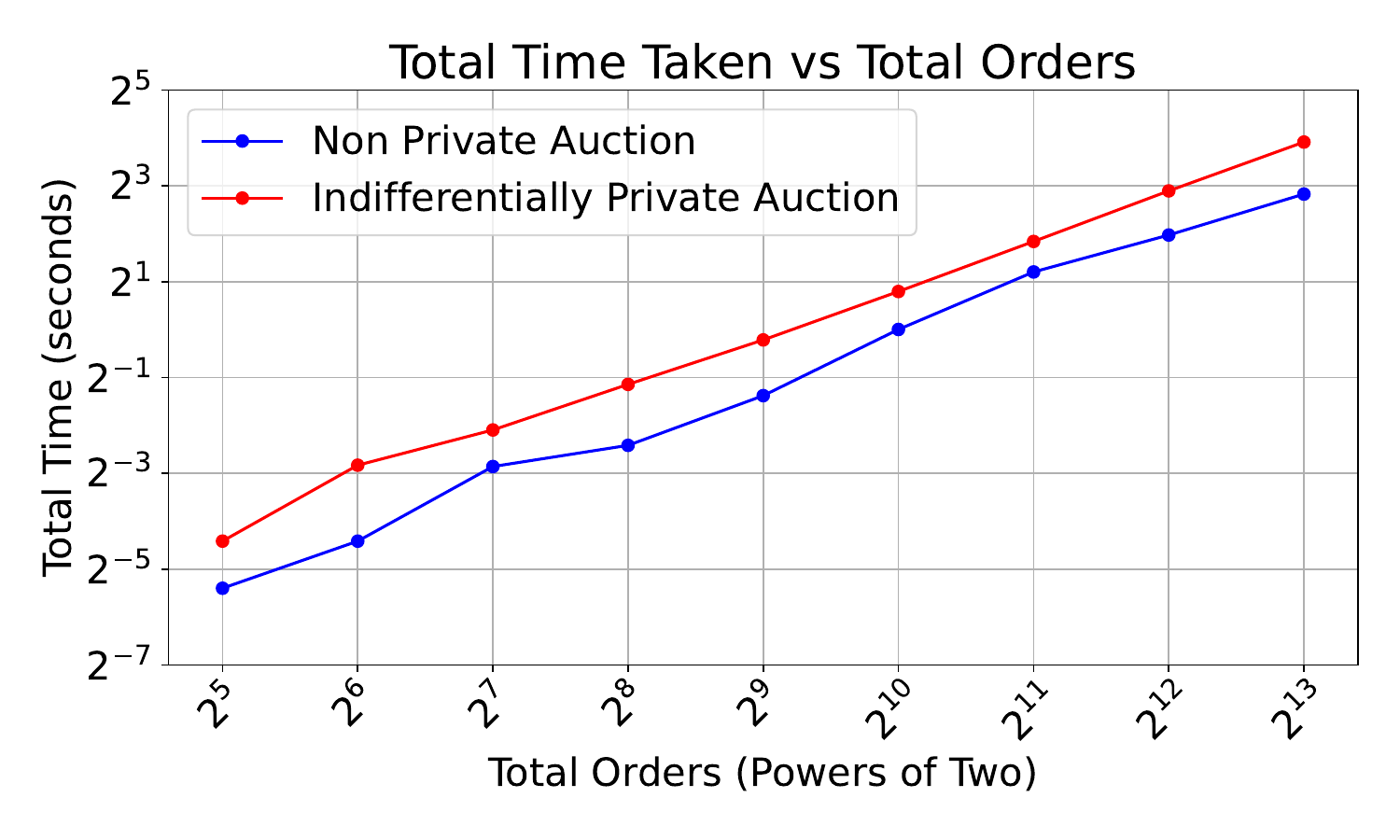}  
        \caption{\textnormal{Very Wide Network setting.}}
        \label{fig:world}
    \end{subfigure}
    \caption{\textbf{Comparison of Time Taken Across Different Settings.}}
    \label{fig:comparison}
\end{figure}

}

From Figure~\ref{fig:comparison}, we can see that the total time taken using indifferential privacy auction does not deviate greatly from the non-private real world implementation. In Table~\ref{tab:comparison} we also mention the orders per second for two different scenarios.

\noindent \textbf{Performance comparison with FHE implementation.}
Considering throughput as the number of orders per second, FHE-based solutions can only process around 0.07 orders per second, as shown in Table 4 of~\cite{MazloomDPB23} for 40 orders, using threshold fully homomorphic encryption (tFHE) implemented via the GPU FHE library of~\cite{BalchDP20}. Moreover, the actual runtime is even longer, as Table 4 does not account for the additional overhead from threshold decryption and key management, which increases with the number of clients (as highlighted in Table 2 of~\cite{AsharovBPV20}).

For context, according to~\cite{AsharovBPV20}, running a comparison on encrypted data for a large amount of orders, $2^{13}$ orders, takes $7.6$ seconds on CPU, while our system processes and matches all orders in $13.32$ seconds in total (see Table~\ref{tab:comparison}), not just performing a comparison on a single order. The FHE-based solution requires at least one comparison operation per matching attempt, adding to its computational burden. Furthermore, the FHE-based timings from~\cite{MazloomDPB23} are measured on GPUs; actual performance on CPUs would be even slower, whereas our experiments are conducted entirely on CPUs.

This illustrates that FHE solutions are an overkill, even with the benefit of GPU acceleration, while our approach offers a practical solution. Last but not least, as shown in Table 4 of~\cite{MazloomDPB23}, the latest multi-operator solution, based on secure multiparty computation, is capable of processing only 26 orders per second, even on high-performance hardware. 

\begin{table}[hbt!]
    \centering
    \begin{tabular}{|c|c|c|}
        \hline
        \textbf{Orders} & \textbf{Non-Private Time (secs)} & \textbf{IDP Time (secs)} \\
        \hline
        40             & 0.019698                        & 0.046887                 \\
        \(2^{13}\)     & 3.887790                        & 13.328184                \\
        \hline
    \end{tabular}
    \caption{Comparison of Non-Private and IDP total running times for a small (40) and large size ($2^{13}$) of orders. The throughput (orders per second) for the non-private case is $2031$ and $2107$ and for the indifferential private case is $853$ and $615$ for 40 and $2^{13}$ orders, respectively.}
    \label{tab:comparison}
\end{table}

\noindent \textbf{Scalability with larger datasets.}
In the  experiments above, we compared the run-times of our protocol with a real-world auction protocol for up to 1,024 clients, each submitting 8 orders (resulting in a total of $2^{13}$ orders). The parameters used in our benchmarks reflect realistic scenarios based on data from U.S. dark pools. We have also extended our experiments to include scenarios where both the number of orders per client and the total number of clients double incrementally, scaling up to a total of 262K orders. From Figure~\ref{fig:scalable}, we can see that the differentially private protocol scales linearly with the increased workload. 

\begin{figure}[H]
    \centering
    \begin{subfigure}[t]{0.45\textwidth}  
        \centering
        \includegraphics[width=\textwidth]{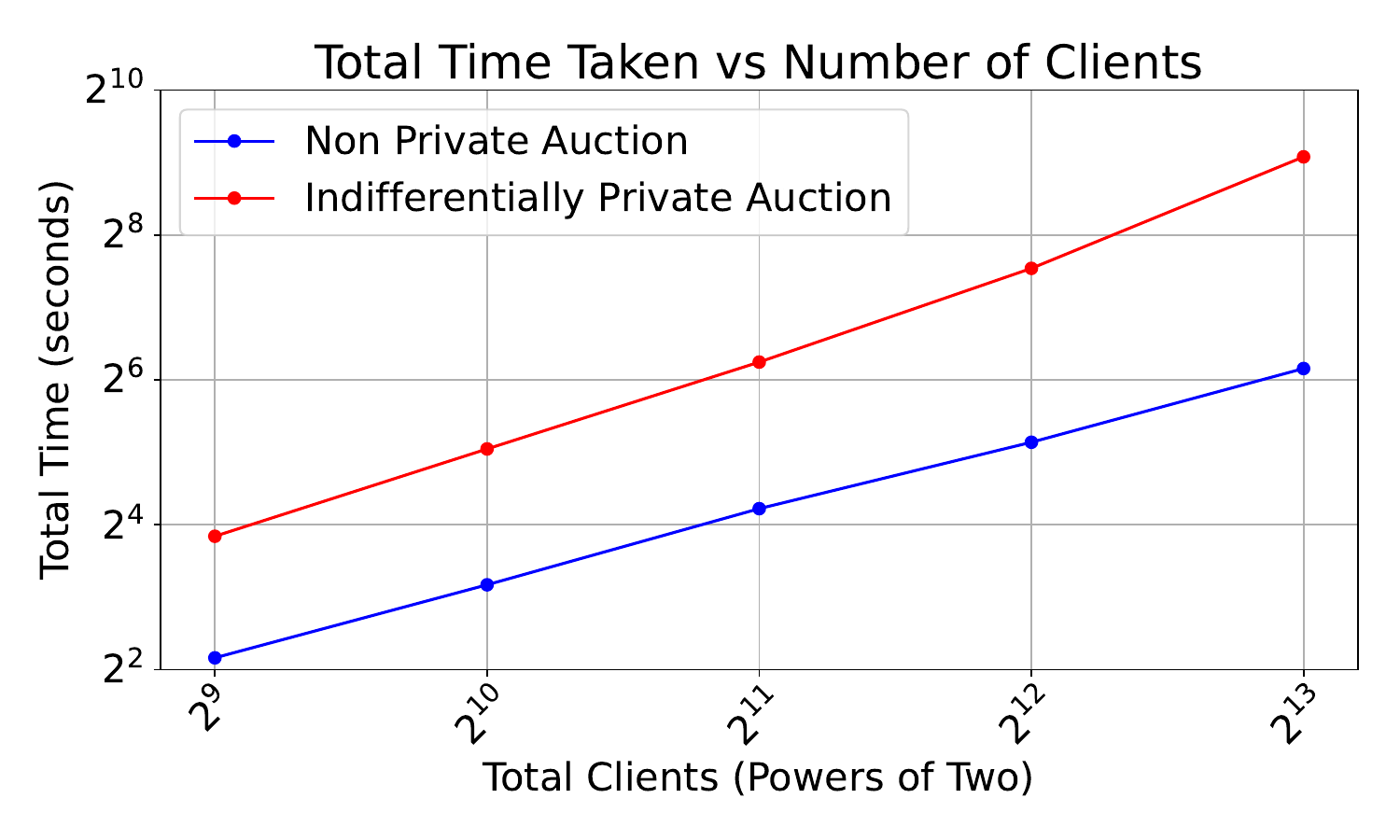}  
        \caption{\textnormal{Runtimes in seconds when number of clients increases and the orders per client is fixed to 8}}
        \label{fig:client}
    \end{subfigure}
    \hspace{1em} 
    \begin{subfigure}[t]{0.45\textwidth}  
        \centering
        \includegraphics[width=\textwidth]{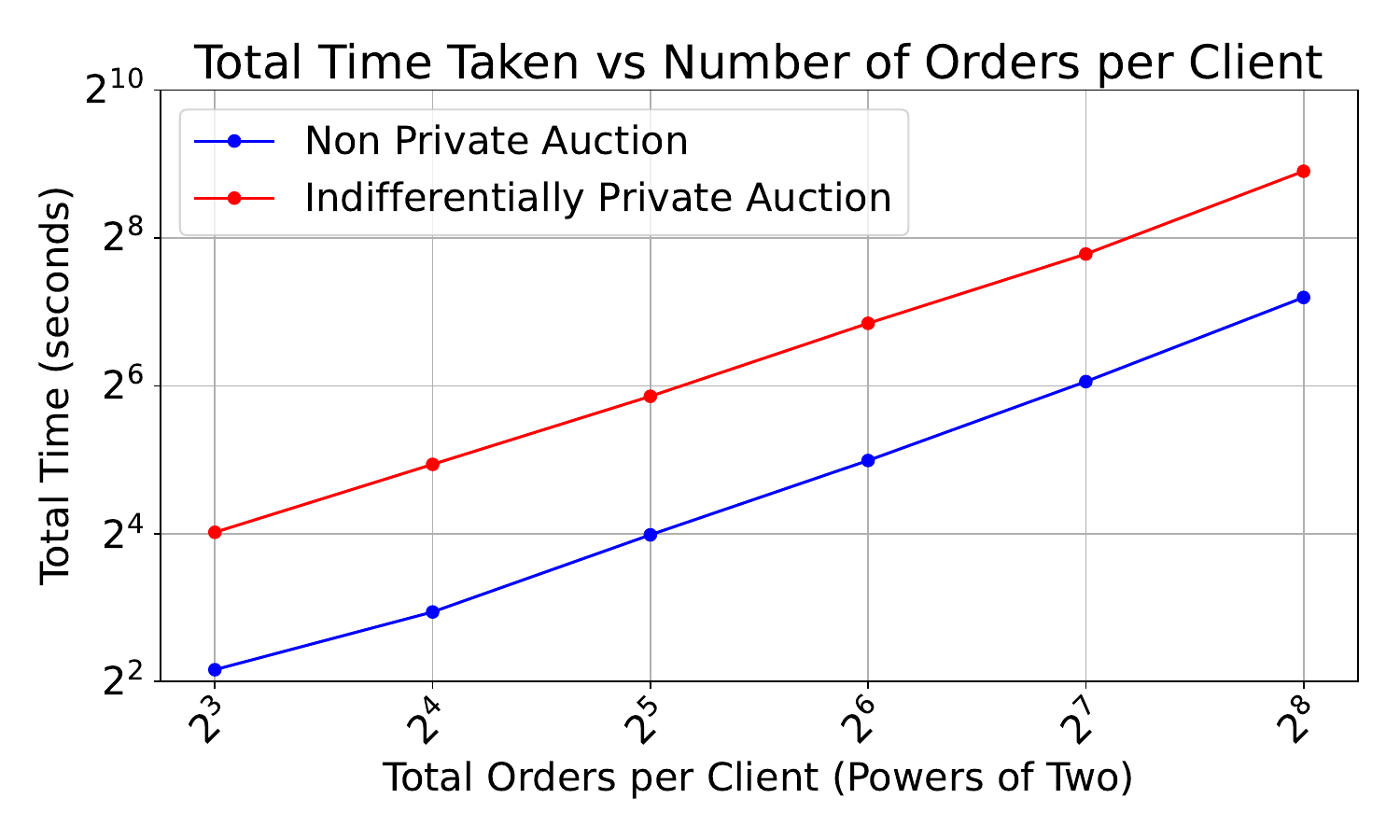}  
        \caption{\textnormal{Runtimes in seconds when the number of orders per client increases and the total number of clients is fixed to 1024}}
        \label{fig:order}
    \end{subfigure}
    \vspace{1em} 
    \caption{\textbf{Extended runtime tables for larger number of clients and orders}}
    \label{fig:scalable}
\end{figure}

\section{Conclusion}
In this work, we addressed the limitations of existing privacy-preserving auctions in high-frequency trading, such as dark pools, where auctioneers can be untrustworthy. Previous methods, like fully homomorphic encryption, were impractical due to their overhead. Our approach, based on the new notion of Indifferential Privacy, provides an efficient, privacy-preserving continuous double auction that enables maximum matching while minimizing risks. This makes our system a practical and secure alternative addressing both performance and security concerns in modern trading.

\section*{Acknowledgements.}
T-H. Hubert Chan was partially supported by the Hong Kong RGC grants
17201823 and 17202121. 
This paper was prepared for informational purposes by the Artificial Intelligence Research group of JPMorgan Chase \& Co.~and its affiliates ("JP Morgan'') and is not a product of the Research Department of JP Morgan. JP Morgan makes no representation and warranty whatsoever and disclaims all liability, for the completeness, accuracy or reliability of the information contained herein. This document is not intended as investment research or investment advice, or a recommendation, offer or solicitation for the purchase or sale of any security, financial instrument, financial product or service, or to be used in any way for evaluating the merits of participating in any transaction, and shall not constitute a solicitation under any jurisdiction or to any person, if such solicitation under such jurisdiction or to such person would be unlawful.

\clearpage

\bibliographystyle{alpha}
\bibliography{sample}

\end{document}